\documentclass[conference]{IEEEtran}
\IEEEoverridecommandlockouts
\usepackage{cite}
\usepackage{amsmath,amssymb,amsfonts}
\usepackage{algorithmic}
\usepackage{graphicx}
\usepackage{textcomp}
\usepackage{xcolor}
\usepackage{comment}
\usepackage{mathptmx}
\usepackage{amsthm}
\newtheorem{thm}{Theorem}

\newtheorem{prop}[thm]{Proposition}
\newtheorem{cor}{Corollary}

\theoremstyle{definition}
\newtheorem{defn}{Definition}

\def\BibTeX{{\rm B\kern-.05em{\sc i\kern-.025em b}\kern-.08em
    T\kern-.1667em\lower.7ex\hbox{E}\kern-.125emX}}
\begin{document}

\title{Database Matching Under Column Deletions\\
\thanks{This work is supported by NYU WIRELESS Industrial Affiliates and National Science Foundation grant CCF-1815821.}}

\author{\IEEEauthorblockN{Serhat Bakırtaş}
\IEEEauthorblockA{\textit{Dept. of Electrical and Computer Engineering} \\
\textit{New York University}\\
NY, USA \\
serhat.bakirtas@nyu.edu}
\and 
\IEEEauthorblockN{Elza Erkip}
\IEEEauthorblockA{\textit{Dept. of Electrical and Computer Engineering} \\
\textit{New York University}\\
NY, USA \\
elza@nyu.edu}
}

\maketitle

\begin{abstract}
De-anonymizing user identities by matching various forms of user data available on the internet raises privacy concerns. A fundamental understanding of the privacy leakage in such scenarios requires a careful study of conditions under which correlated databases can be matched. Motivated by synchronization errors in time indexed databases, in this work, matching of random databases under random column deletion is investigated. Adapting tools from information theory, in particular ones developed for the deletion channel, conditions for database matching in the absence and presence of deletion location information are derived, showing that partial deletion information significantly increases the achievable database growth rate for successful matching. Furthermore, given a batch of correctly-matched rows, a deletion detection algorithm that provides partial deletion information is proposed and a lower bound on the algorithm's deletion detection probability in terms of the column size and the batch size is derived. The relationship between the database size and the batch size required to guarantee a given deletion detection probability using the proposed algorithm suggests that a batch size growing double-logarithmic with the row size is sufficient for a nonzero detection probability guarantee.
\end{abstract}

\section{Introduction}
\label{sec:introduction}

In the last decade, especially with the proliferation of smart devices and the rise of social media, there has been a boom in data collection. As the collection of potentially sensitive personal data by companies and governments has increased, so has the risk of privacy leakage due to sale and publication of collected data. The privacy concerns over the publication of the anonymized data  have been articulated recently where \cite{naini2015you,datta2012provable,narayanan2008robust,sweeney1997weaving,takbiri2018matching} have shown that anonymization is not sufficient on its own to prevent privacy leakage. In particular, these works devise practical attacks and use them on real data to match anonymized database with publicly available user information. While these attacks work efficiently on real data, \cite{naini2015you,datta2012provable,narayanan2008robust,sweeney1997weaving,takbiri2018matching} do not suggest a fundamental understanding of what kind of data is vulnerable to privacy attacks. 

More recently matching of correlated databases have been rigorously investigated in \cite{shirani8849392} and \cite{cullina}. In \cite{shirani8849392}, Shirani \textit{et al.} developed a matching scheme based on joint typicality and derived necessary and sufficient conditions on the database growth rate for realiable matching using an extension of Shannon-McMillan-Breiman Theorem and Fano's inequality. In \cite{cullina}, Cullina \textit{et al.} introduced \textit{cycle mutual information} as a new correlation metric and derived sufficient conditions for a successful matching and a converse result.

In this paper, we further the study of database matching by considering random column deletions. To motivate column deletions, consider the following scenario illustrated in Figure \ref{fig:intro}: We have access to two anonymized databases containing time-indexed transactions of a set of users made respectively through a bank account and a credit card associated with it, where the time indices don't necessarily match, i.e. there may be synchronization errors. By matching these users across these correlated databases, an attacker could gain useful information on user spending profiles or the bank can detect a potentially fraudulent activity.

\begin{figure}[t]
\centerline{\includegraphics[width=0.50\textwidth]{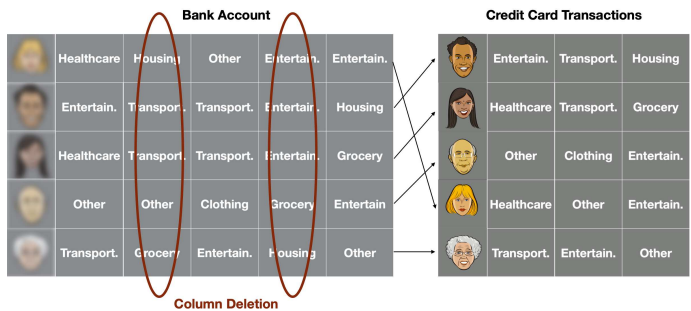}}
\caption{An illustrative example of database matching under column deletions. Each row corresponds to a user and each database entry is the type of a transaction.}
\label{fig:intro}
\end{figure}

We model the above example as a database matching problem where the goal is to match the corresponding rows across databases such that the probability of mismatch goes to zero as the number of attributes in the database (number of columns) grows to infinity. The two databases are assumed to have the same number of users (rows) and are generated according to a bivariate stochastic process as in \cite{shirani8849392}. Different than \cite{shirani8849392}, the second database suffers from \textit{column deletion}. The indices of the deleted columns are not known due to synchronization errors similar to the deletion channel model \cite{cheraghchi2020overview}. We also assume availability of partial deletion location information, where a subset of deleted column indices are known.

 Our goal is to investigate sufficient conditions for the successful matching of rows under column deletions, in the presence of partial deletion location information. We first derive conditions on the database size, deletion probability and amount of partial deletion location information for successful matching. In many practical problems, rather than partial deletion location information, a batch of already-matched rows, which we call \textit{seeds} may be available. Given such a batch, we propose an algorithm which detects deleted columns by exploiting the fact that the same set of columns is deleted in each row. Furthermore, we present a lower bound to this algorithm's deletion detection probability in terms of the column size of the database, $n$, and the row size of the correctly-matched batch, $B$. In turn, we investigate the relation between the row size of the database, $m$ and $B$, for a given performance guarantee in terms of deletion detection probability. We argue that as long as $B$ grows faster than $\log\log m$, all deleted columns can be detected, pointing that even a small seed size may help with matching.

The organization of this paper is as follows: Section \ref{sec:problemformulation} contains the formulation of the problem. In Section \ref{sec:achievabledbgrowthrates}, results on sufficient conditions for the successful database matching are presented. In Section \ref{sec:deletionlocationextraction}, for a given batch of correctly-matched rows, an algorithm for deletion detection is proposed and the relation between the detection probability of the algorithm, the column size and the size of the batch are investigated. Finally, in Section \ref{sec:conclusion} the results are discussed.
\newline
\newline
\noindent{\em Notation:} We denote the set of integers $\{1,2,...,n\}$ as $[n]$,  databases with calligraphic letters, (e.g. $\mathcal{C}$), random vectors with bold uppercase letters. For a set of indices $I_D=\{i_1,i_2,...,i_d\}\subseteq [n]$, we denote the vector $(X_1,...,X_{i_1-1},X_{i_1+1},...,X_{i_2-1},X_{i_2+1},...,X_{i_d-1},X_{i_d+1},...)$ of length $n-d$ with $\mathbf{X}([n]\setminus I_D)$.

\section{Problem Formulation}
\label{sec:problemformulation}

We use the following definitions, some of which are taken from \cite{shirani8849392} to formalize our problem.
\begin{defn}{(Unlabeled Database)}
An $(m,n,p_{X})$ \textit{unlabeled random database} is a randomly generated $m\times n$ matrix $\mathcal{C}=\{X_{i,j}\in\mathfrak{X}^{m\times n}\}$ with \textit{i.i.d.} entries drawn according to the distribution $p_X$ from a discrete alphabet $\mathfrak{X}$. The $i$\textsuperscript{th} row $\mathbf{X}_i$ of $\mathcal{C}$ is said to correspond to user $i$. Here $m$ and $n$ represent the number of users and the number of attributes, respectively.
\end{defn}

\begin{defn}{(Column Deletion Pattern)}
Column deletion pattern $\mathbf{D}^n=\{D_1,D_2,...,D_n\}$ is a random vector with \textit{i.i.d.} Bern$(\delta)\in\{0,1\}$ entries, independent of $\mathcal{C}^{(1)}$, $D_i=1$ indicating that the $i$\textsuperscript{th} column is deleted. The Bernoulli parameter $\delta$ is called the \textit{column deletion probability}.
\end{defn}

\begin{defn}{(Column Deleted Labeled Database)}\label{defn:corrdeldb}
Let $\mathcal{C}^{(1)}$ be an $(m,n,p_{X})$ unlabeled database. Let $\mathbf{D}^n$ be the column deletion pattern, $\boldsymbol{\Theta}$ be a permutation of $[m]$. Given $\mathcal{C}^{(1)}$ and $\mathbf{D}^n$, the pair $(\mathcal{C}^{(2)},\boldsymbol{\Theta})$ is called the \textit{column deleted labeled database} if $\mathbf{R}_i^{(1)}$ and $\mathbf{R}_i^{(2)}$ have the following relation:
$$\mathbf{R}_i^{(2)}=\left\{
    \begin{array}{ll}
      \mathbf{E} , &  \text{if } D_i=1\\
      \boldsymbol{\Theta}\circ \mathbf{R}_i^{(1)} & \text{if } D_i=0
    \end{array} 
\right.$$
where $\mathbf{R}_i^{(j)}$ denotes the $i$\textsuperscript{th} column of the database $\mathcal{C}^{(j)}$ and $\mathbf{R}_i^{(2)}=\mathbf{E}$ corresponds to all entries of $\mathbf{R}_i^{(2)}$ being the empty string. Therefore, given a deletion pattern $\mathbf{D}^n$, the column size of $\mathcal{C}^{(2)}$ is $
\sum\limits_{i=1}^n D_i$, which is a $Binomial(n,1-\delta)$ random variable, independent of the database entries. 

For the databases in Definition \ref{defn:corrdeldb}, the $i$\textsuperscript{th} row $\mathbf{Y}_i$ of $\mathcal{C}^{(2)}$ is said to correspond to the user $\boldsymbol{\Theta}^{-1}(i)$. The rows $\mathbf{X}_{i_1}$ and $\mathbf{Y}_{i_2}$ are said to be \textit{matching rows}, if $\boldsymbol{\Theta}(i_1)=i_2$, where $\boldsymbol{\Theta}$ is called the \textit{labeling function}.
\end{defn}
Notice $\mathcal{C}^{(2)}$ is obtained by shuffling $\mathcal{C}^{(1)}$ with $\boldsymbol{\Theta}$ followed by column deletion, and there is no noise on the retained entries, similar to the deletion channel model \cite{cheraghchi2020overview}.

\begin{defn}{(Deletion Detection Pattern)} 
Given the column deletion pattern $\mathbf{D}^n$, the \textit{column deletion detection pattern} $\mathbf{A}^n=\{A_1,A_2,...,A_n\}$ is a random vector independent of $\mathcal{C}^{(1)}$, with independent entries having the following conditional distribution:
$$P(A_i=1|D_i)=\alpha \mathbf{1}_{[D_i=1]},\quad \forall i\in[n]$$
where $\mathbf{1}_\epsilon$ is the indicator function of event $\epsilon$. The parameter $\alpha\in[0,1]$ is called the \textit{deletion detection probability}.
\end{defn}

\begin{defn}{(Database Growth Rate)}
The \textit{database growth rate} $R$ of an $(m,n,p_X)$ unlabeled database is defined as $$R=\lim\limits_{n\to\infty} \frac{1}{n}\log_2 m$$
\end{defn}
\begin{defn}{(Successful Matching Scheme)}
Given a deletion detection pattern $\mathbf{A}^n$, a \textit{matching scheme} is a sequence of mappings $s_n: (\mathcal{C}^{(1)},\mathcal{C}^{(2)})\to \hat{\boldsymbol{\Theta}}_n $ where $\hat{\boldsymbol{\Theta}}_n\in [m]^m$ is the estimate of the correct permutation $\boldsymbol{\Theta}_n$. The scheme is \textit{successful} if 
$$P(\boldsymbol{\Theta}_n(I)=\hat{\boldsymbol{\Theta}}_n(I))\to 1  \text{ as }n\to\infty$$
where the index $I$ is drawn uniformly from $[m]$. Here the dependence of $\hat{\boldsymbol{\Theta}}_n$ on $\mathbf{A}^n$ is omitted for brevity. 
\end{defn}
\begin{defn}{(Achievable Database Growth Rate)}\label{defn:achievable}
Given a database probability distribution $p_X$, column deletion probability $\delta$ and deletion detection probability $\alpha$, a database growth rate $R$ is said to be achievable if for any pair of databases $(\mathcal{C}^{(1)},\mathcal{C}^{(2)})$ with database growth rate $R$, 
 there exists a successful matching scheme.
\end{defn}

\section{Achievable Database Growth Rates}
\label{sec:achievabledbgrowthrates}
In this section, our goal is to derive achievable database growth rates as in Definition \ref{defn:achievable} and associated matching schemes.

In the following theorem, we consider the following matching strategy: We first discard all the deleted columns of $\mathcal{C}^{(1)}$ that are detected, exploiting the fact that all the rows have the same deletion pattern. Then, we use a row matching scheme following \cite{shirani8849392} and \cite{diggavi1603788}. Our strategy matches each row separately and does not use the fact that each row has identical deletion pattern. In Section \ref{sec:deletionlocationextraction} we show that exploiting the deletion pattern across rows can in fact be very beneficial. Furthermore, it should be emphasized that one could perform the matching at the database level to potentially achieve higher database growth rates.

\begin{thm}\label{thm:achievablealpha}
Consider an unlabeled database generated according to $p_X$ with alphabet $\mathfrak{X}$ and a column deletion probability $\delta<1-\frac{1}{|\mathfrak{X}|}$. For a deletion detection probability $\alpha$, any database growth rate
\begin{align*}
    R<\Big[(1-\alpha\delta)&\left(H(X)-H_b\left(\frac{1-\delta}{1-\alpha\delta}\right)\right)\\
        &-(1-\alpha)\delta \log(|\mathfrak{X}|-1)\Big]^+
\end{align*}
    is achievable, where $H$,$H_b$ and $[.]^+$ denote the entropy, the binary entropy, and the positive part functions respectively.
\end{thm}

Note that one could rearrange the terms on the right-hand side as the following:
\begin{align*}
    \Big[(1-\delta)H(X)&-(1-\alpha)\delta\left(\log(|\mathfrak{X}|-1)-H(X)\right)\\&-(1-\alpha\delta)H_b\left(\frac{1-\delta}{1-\alpha\delta}\right)\Big]^+
\end{align*}
where the term $(1-\delta) H(X)$ corresponds to achievable rate in the presence of full deletion location information ($\alpha=1$), the second term is the penalty due to a potentially low $H(X)$ causing $\mathcal{C}^{(1)}$ to have similar entries in each row and thus increasing the error probability, and the last term represents the penalty paid for the lack of deletion location information. Since the penalty terms decrease with $\alpha$, intuitively Theorem \ref{thm:achievablealpha} states that as more deleted columns are detected, the matching becomes easier due to lower dimensionality of the search space.

\begin{proof}

Let $\mathbf{D}^n$ and $\mathbf{A}^n$ be the deletion and the deletion detection patterns, respectively. Let $K=n-\sum\limits_{i=1}^n D_i$ be the random variable corresponding to the number of columns in $\mathcal{C}^{(2)}$. Then, for any $\Tilde{\epsilon}>0$ we have 
$$P\left(\left|\frac{K}{n}-(1-\delta)\right|>\Tilde{\epsilon}\right)\to 0\text{ as }n\to\infty$$
Choose $k=\lfloor n(1-\delta-\Tilde{\epsilon})\rfloor$. Note that for any $K\ge k$, $\frac{n-K}{n}\le \delta+\Tilde{\epsilon}$ as $n\to\infty$. Denoting the probability that $K<k$ by $\kappa_n$ and using the Law of Large Numbers, we have $\kappa_n\to 0$ as $n\to\infty$. 

Now, let $I_A$ be the set of detected deletion indices, and $A=|I_A|=\sum\limits_{i=1}^n A_i$. Then, for any $\hat{\epsilon}>0$ we have
$$P\left(\left|\frac{A}{n-k}-\alpha\right|>\hat{\epsilon}\right)\to 0\text{ as }n\to\infty$$
Choose $a=\lfloor (n-k)(\alpha-\hat{\epsilon})\rfloor$. Note that for any $A\ge a$, $\frac{A}{n-k}\ge \alpha-\hat{\epsilon}$ as $n\to\infty$. Denoting the probability that $A<a$ by $\mu_n$, using the Law of Large Numbers, we have $\mu_n\to 0$ as $n\to\infty$. 

Let $A_\epsilon^{(n-a)}(X)$ be the $\epsilon$-typical set associated with $p_X$ with parameter $n-a$. Consider the following matching scheme: First, we discard all columns whose index belongs to $I_A$, since these columns are known to be deleted. Given a row $\mathbf{Y}_{j_1}$ of $\mathcal{C}^{(2)}$, we match the row $\mathbf{X}_{i_1}$ of $\mathcal{C}^{(1)}$ assigning $\hat{\boldsymbol{\Theta}}^{-1}(j_1)=i_1$, if $\mathbf{X}_{i_1}([n]\setminus I_A)$ contains $\mathbf{Y}_{j_1}$,  $\mathbf{X}_{i_1}([n]\setminus I_A)\in A_\epsilon^{(n-a)}(X)$ and there is no other row $\mathbf{X}_{i_2}^n$ of $\mathcal{C}^{(1)}$ with $\mathbf{X}_{i_2}([n]\setminus I_A)\in A_\epsilon^{(n-a)}(X)$ containing $\mathbf{Y}_{j_1}$ potentially in a non-contiguous way. We say that in that case no \textit{collision} occurs. If any of these steps fail, we declare an error.

In addition, the matching scheme only considers $K\ge k,A\ge a$ and otherwise declares an error. Since additional columns in $\mathcal{C}^{(2)}$ and additional detected deleted columns would decrease the collision probability, we have
$$P(\text{collision}|K\ge k,A\ge a)\le P(\text{collision}|K=k,A=a)$$

Denote the pairwise collision probability between $\mathbf{X}_1$ and $\mathbf{X}_i$, by $P_{col,i}$. Therefore given the correct labeling for $\mathbf{Y}\in\mathcal{C}^{(2)}$ is $\mathbf{X}_1\in\mathcal{C}^{(1)}$, the probability of error can be bounded as
\begin{align}
    P_e 
    &\le \sum\limits_{i=2}^{2^{n R}} P_{col,i}+\epsilon+\kappa_n+\mu_n\notag\\
    &\le 2^{n R} P_{col,2}+\epsilon+\kappa_n+\mu_n \label{eq1:Pe}
\end{align}
where we used that the rows are \textit{i.i.d.} and $P_{col,i}=P_{col,2}$.
Let $F(n,k,|\mathfrak{X}|)$ denote the number of $|\mathfrak{X}|$-ary sequences of length $n$, which contain a fixed $|\mathfrak{X}|$-ary sequence of length $k$. Since $\frac{k}{n}\ge 1-\delta-\Tilde{\epsilon}$ and $\delta\le1-\frac{1}{|\mathfrak{X}|}$, we have $\frac{k}{n}\ge \frac{1}{|\mathfrak{X}|}-\Tilde{\epsilon}$. Then from \cite{chvatal1975longest} and \cite{cover2006elements} (Chapter 11) we have the following upper bound for $k\ge\frac{n}{|\mathfrak{X}|}$:
\begin{align*}
    F(n,k,|\mathfrak{X}|)
    &\le n 2^{n H_b\left(k/n\right)} (|\mathfrak{X}|-1)^{n-k}
\end{align*}

Let $T(\mathbf{y},I_A)=\{\mathbf{x}\in\mathfrak{X}^n|\mathbf{x}([n]\setminus I_A)\in A_\epsilon^{(n-a)}$ contains $\mathbf{y}\}$ and $\mathbf{y}$ be the row of $\mathcal{C}^{(2)}$ matching with the row $\mathbf{X}_1$ of $\mathcal{C}^{(1)}$. It is clear that $|T(\mathbf{y},I_A)|\le F(n-a,k,|\mathfrak{X}|)$. Also for any $\mathbf{x}\in T(\mathbf{y},I_A)$, since $\mathbf{x}([n]\setminus I_A)\in A_\epsilon^{(n-a)}$ we have
$$p(\mathbf{x})\le 2^{-(n-a)(H(X)-\epsilon)}$$

Since the rows are \textit{i.i.d.} we have 
$$P(\mathbf{X}_2 \in T(\mathbf{y},I_A)|\mathbf{X}_1 \in T(\mathbf{y},I_A))=P(\mathbf{X}_2 \in T(\mathbf{y},I_A))$$
Then $P_{col,2}$ can be bounded as

\begin{align}
    P_{col,2} &= P(\mathbf{X}_2 \in T(\mathbf{y},I_A))\notag\\
    &=\sum\limits_{\mathbf{x}\in T(\mathbf{y},I_A)} p(\mathbf{x})\notag\\
    &\le \sum\limits_{\mathbf{x}\in T(\mathbf{y},I_A)} 2^{-(n-a)(H(X)-\epsilon)}\notag\\
    &\le 2^{-(n-a)(H(X)-\epsilon)} F(n-a,k,|\mathfrak{X}|)\notag\\
    &\le (n-a) 2^{-(n-a)\left(H(X)-\epsilon-H_b\left(\frac{k}{n-a}\right)\right)}(|\mathfrak{X}|-1)^{n-a-k}\label{eq:pcol}
\end{align}

Combining \eqref{eq:pcol} with \eqref{eq1:Pe}, we have
\begin{align*}
    P_e
    &\le (n-a) 2^{-n\left[(1-\frac{a}{n})\left(H(X)-\epsilon-H_b\left(\frac{k}{n-a}\right)\right)-R\right]}(|\mathfrak{X}|-1)^{n-a-k}\\
    &\qquad +\epsilon+\kappa_n+\mu_n\\
    &\le \epsilon
\end{align*}
 as $n\to\infty$ if 
 \begin{align*}
     R<\Big[\left(1-\frac{a}{n}\right)&\left(H(X)-\epsilon-H_b\left(\frac{k}{n-a}\right)\right)\\
        &-\left(1-\frac{a}{n-k}\right)\frac{n-k}{n} \log(|\mathfrak{X}|-1)\Big]^+\\
 \end{align*}
Thus, we can argue that any rate $R$ satisfying
     \begin{align*}
     R<\Big[(1-\alpha\delta)&\left(H(X)-H_b\left(\frac{1-\delta}{1-\alpha\delta}\right)\right)\\
        &-(1-\alpha)\delta\log(|\mathfrak{X}|-1)\Big]^+
 \end{align*}
 is achievable by taking $\epsilon$, $\Tilde{\epsilon}$ and $\hat{\epsilon}$ small enough.
\end{proof}

\begin{cor}{(No Deletion Location Information)}\label{cor:nodel}
In the absence of deletion location information ($\alpha=0$), any database growth rate $R$ satisfying
$$R<\left[H(X)-H_b(\delta)-\delta \log(|\mathfrak{X}|-1)\right]^+$$
is achievable.
\end{cor}

\begin{cor}{(Full Deletion Location Information)}\label{cor:fulldel}
In the presence of full deletion location information ($\alpha=1$), any database growth rate $R$ satisfying
$$R<(1-\delta) H(X)$$
is achievable.
\end{cor}

The achievable rate as a function of the deletion probability for different the deletion detection probabilities is illustrated in Figure \ref{fig:Rvsdel}.

Note that since the deletion pattern across rows is not exploited in Theorem \ref{thm:achievablealpha}, Corollary \ref{cor:nodel} is closely related to the deletion channel rate \cite{diggavi1603788}, while Corollary \ref{cor:fulldel} is related to the erasure channel capacity. However, in contrast to the channel capacity results, in the database matching problem, the database distribution $p_X$ is fixed and cannot be optimized.

\begin{figure}[t]
\centerline{\includegraphics[width=0.40\textwidth,height=0.26\textheight]{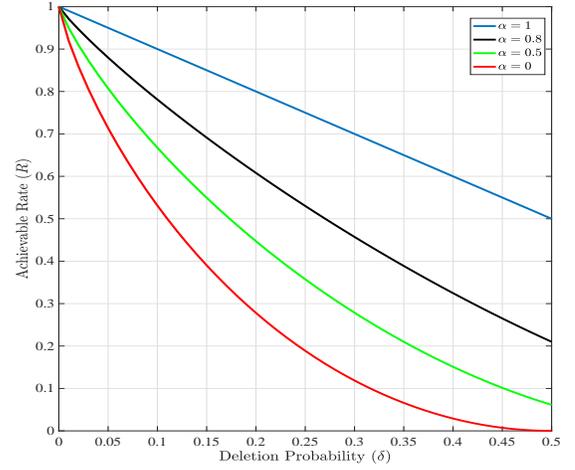}}
\caption{Achievable database growth rate ($R$) vs. deletion probability ($\delta$) for different deletion detection probabilities ($\alpha$), when $X\sim Bernoulli(\frac{1}{2})$. Notice that for $\delta\approx0.4$ there is a twenty-fold difference between the achievable rates in the presence ($\alpha=1$) and absence ($\alpha=0$) of the deletion location information, showing the significance of deletion detection, for fairly large $\delta$.}
\label{fig:Rvsdel}
\end{figure}

\section{Deletion Detection}
\label{sec:deletionlocationextraction}
In Section \ref{sec:achievabledbgrowthrates}, we assumed a given deletion detection probability $\alpha$ and found a corresponding achievable database growth rate. However, in practice one may not have such a partial deletion location information. One could have a correctly-matched set of rows as \textit{seeds} (\cite{shirani2017seeded,fishkind2019seeded}). In this section, we assume we have access to a seed of $B$ correctly-matched rows of databases $\mathcal{C}^{(1)}$ and $\mathcal{C}^{(2)}$, denoted by $\mathcal{D}^{(1)}$ and $\mathcal{D}^{(2)}$, respectively. Note that having access to a batch of correctly-matched rows does not immediately reveal the deletion locations because many different deletion patterns may lead to the same row in $\mathcal{C}^{(2)}$. We propose an algorithm which extracts deletion location information from $B$ given seeds by exploiting the fact that the deletion occurs columnwise. Then we derive a lower bound on the deletion detection probability of our algorithm. 

Given two sets of correctly-matched rows $\mathcal{D}^{(1)}$ and $\mathcal{D}^{(2)}$, let $S(\mathcal{D}^{(1)},\mathcal{D}^{(2)})$ denote the number of column deletion patterns through which $\mathcal{D}^{(2)}$ can be obtained from $\mathcal{D}^{(1)}$. Here the counting function $S$ is an extension of a similar counting function, described in \cite{mitzenmacher2009survey}, to the columnwise deletion case.

A simple application of Bayes' theorem gives us the following proposition:
\begin{prop}\label{prop:posterior}
Let $I_D\subset [n]$ be the set of deletion indices. Given a batch of $B$ seeds $\mathcal{D}^{(1)},\mathcal{D}^{(2)}$, the posterior deletion probability  of a column $j\in[n]$ is
$$P(j\in I_D|\mathcal{D}^{(1)},\mathcal{D}^{(2)})=\frac{S(\Tilde{\mathcal{D}}_j^{(1)},\mathcal{D}^{(2)})}{S(\mathcal{D}^{(1)},\mathcal{D}^{(2)})}$$
where $\Tilde{\mathcal{D}}_j^{(1)}$ is obtained by removing the $j$\textsuperscript{th} column of $\mathcal{D}^{(1)}$ and appending the rest of the columns.
\end{prop}

Our proposed algorithm classifies columns into the set of deleted columns, the set of retained columns, and the set of columns where the algorithm fails to make a decision, based on the posterior deletion probabilities given in Proposition \ref{prop:posterior}, calculated for a given batch of $B$ correctly-matched rows.

Let $A_\epsilon^{(B)}$ be the $\epsilon$-typical set associated with $p_X$ with parameter $B$, $K$ be the (random) number of columns in $\mathcal{D}^{(2)}$ and $\mathbf{D}_j$ denote the $j$\textsuperscript{th} column of $\mathcal{D}^{(1)}$. Given a batch of correctly-matched pairs of B rows, we first calculate the posterior probability vector $\mathbf{P}=[P_1,...,P_n]$ from Proposition \ref{prop:posterior}.
We then define the \textit{deletion detection function} $f:\mathfrak{X}^{B\times n}\times \mathfrak{X}^{B\times K}\times [n]\to\{0,1,\text{inc}\}$ by
$$f(\mathcal{D}^{(1)},\mathcal{D}^{(2)},j)=\left\{
    \begin{array}{lll}
      0, &  P_j=0 \text{ and } \mathbf{D}_j\in A_\epsilon^{(B)} \\
      1, & P_j=1 \text{ and } \mathbf{D}_j\in A_\epsilon^{(B)}\\
     \text{inc}, & otherwise 
    \end{array} 
\right.$$
Here $f(\mathcal{D}^{(1)},\mathcal{D}^{(2)},j)=1$ implies that the $j$\textsuperscript{th} column is certainly deleted while $f(\mathcal{D}^{(1)},\mathcal{D}^{(2)},j)=0$ implies that the $j$\textsuperscript{th} column is certainly retained. Otherwise we do not make a decision and denote this inconclusive result by inc.

A lower bound on the performance of the deletion detection function $f$ in terms of the probability of detecting a deleted column is provided in the next theorem.
\begin{thm}\label{thm:alphalowerbound}
For the database matching problem in Section \ref{sec:problemformulation}, assume no partial deletion location information, ($\alpha=0$). Let $\mathcal{D}^{(1)},\mathcal{D}^{(2)}$ be a batch of correctly-matched $B$ rows of the unlabeled database $\mathcal{C}^{(1)}$, and the corresponding column deleted database $\mathcal{C}^{(2)}$. Then $$P(f(\mathcal{D}^{(1)},\mathcal{D}^{(2)},j)=1|j\in I_D)\ge 1-\epsilon-n 2^{-B(H(X)-\epsilon)} (1-\delta)$$
\end{thm}
\begin{proof}
Consider a simpler deletion detection function which decides if a column is deleted or not by looking at the existence of the columns of $\mathcal{D}^{(1)}$ in $\mathcal{D}^{(2)}$. Since no noise is present on the retained columns, if a column is missing from $\mathcal{D}^{(2)}$, this function decides that the column is deleted, otherwise it doesn't make any decision. We define this simpler function as $g:\mathfrak{X}^{B\times n}\times \mathfrak{X}^{B\times K}\times [n]\to\{1,\text{inc}\} \text{ where}$ 
$$g(\mathcal{D}^{(1)},\mathcal{D}^{(2)},j)=\left\{
    \begin{array}{lll}
      1, &  \mathbf{D}_j\text{ is not a column of }\mathcal{D}^{(2)} \\
      &\text{and }\mathbf{D}_j\in A_\epsilon^{(B)}\\
     \text{inc}, & otherwise 
    \end{array} 
\right.$$
Note that the function $f$ focuses on both the order and the existence of the columns of $\mathcal{D}^{(1)}$ in $\mathcal{D}^{(2)}$ whereas $g$ only focuses on the existence. Furthermore, if $\mathbf{D}_j\text{ is not a column of }\mathcal{D}^{(2)}$, one can discard it from $\mathcal{D}^{(1)}$ when counting the number patterns $\mathcal{D}^{(2)}$ occurs columnwise in $\mathcal{D}^{(1)}$. In other words if, $\mathbf{D}_j\text{ is not a column of }\mathcal{D}^{(2)}$, then $$S(\mathcal{D}^{(1)},\mathcal{D}^{(2)})=S(\Tilde{\mathcal{D}}_j^{(1)},\mathcal{D}^{(2)})$$ Thus $g(\mathcal{D}^{(1)},\mathcal{D}^{(2)},j)=1$ implies that $f(\mathcal{D}^{(1)},\mathcal{D}^{(2)},j)=1$.
For brevity, let $\alpha=P(f(\mathcal{D}^{(1)},\mathcal{D}^{(2)},j)=1|j\in I_D)$. Then using the fact that the columns $\mathbf{D}_j$ are \textit{i.i.d.} and the deletion is independent of $\mathcal{D}^{(1)}$, we have the following
\begin{align*}
    1-\alpha &= P(f(\mathcal{D}^{(1)},\mathcal{D}^{(2)},j)\neq1|j\in I_D)\\
     &\le  P(g(\mathcal{D}^{(1)},\mathcal{D}^{(2)},j)\neq1|j\in I_D)\\
        &=  P(\mathbf{D}_j\text{ is a column of }\mathcal{D}^{(2)}|j\in I_D, \mathbf{D}_j\in A_\epsilon^{(B)})\\
        &\qquad P(\mathbf{D}_j\in A_\epsilon^{(B)}) +P(\mathbf{D}_j\notin A_\epsilon^{(B)})\\
        &\le  P(\exists i\neq j, \mathbf{D}_j=\mathbf{D}_i, i\notin I_D|j\in I_D,\mathbf{D}_j\in A_\epsilon^{(B)})+\epsilon\\
        &\le  P(\exists i\neq j, \mathbf{D}_j=\mathbf{D}_i, i\notin I_D|\mathbf{D}_j\in A_\epsilon^{(B)})+\epsilon\\
        &\le \sum\limits_{i=1;i\neq j}^{n} P(\mathbf{D}_j=\mathbf{D}_i |i\notin I_D, \mathbf{D}_j\in A_\epsilon^{(B)}) P(i\notin I_D)+\epsilon\\
        &=\sum\limits_{i=1;i\neq j}^{n} P(\mathbf{D}_i=\mathbf{D}_j|\mathbf{D}_j\in A_\epsilon^{(B)}) P(i\notin I_D)+\epsilon\\
    &\le\sum\limits_{i=1;i\neq j}^{n} 2^{-B(H(X)-\epsilon)} (1-\delta)+\epsilon\\
    &\le n 2^{-B(H(X)-\epsilon)} (1-\delta)+\epsilon
\end{align*}
which completes the proof.
\end{proof}

\begin{cor}
To guarantee $P(f(\mathcal{D}^{(1)},\mathcal{D}^{(2)},j)=1|j\in I_D)\ge \alpha$, a batch size of $B\ge \frac{1}{H(X)}\log \left(n\frac{1-\delta}{1-\alpha}\right)$ is needed. This suggests that a seed size of $O(\log n)=O(\log\log m)$ ensures a non-zero deletion detection probability $\alpha$. Furthermore if $B$ grows slower than $\log n$, the lower bound becomes trivial.
\end{cor}

\begin{cor}\label{rem:Brelatedton}
If $B=\omega(\log n)=\omega(\log\log m)$, for large $n$, we have $P(f(\mathcal{D}^{(1)},\mathcal{D}^{(2)},j)=1|j\in I_D)\ge 1-\epsilon$ .
\end{cor}

In Theorem \ref{thm:achievablealpha}, we assumed that detection of each deleted column is independent of the remaining deleted columns. However, the deletion detection discussed in this section does not necessarily lead to independence. In fact, no algorithm which extracts the deletion locations from databases directly can lead to an \textit{i.i.d.} detection process. For example, consider two adjacent columns with identical entries, both being deleted. We can detect deletion of either both columns or none.

\section{Conclusion}
\label{sec:conclusion}
In this work, we have studied a database matching problem under random column deletions. We have found an achievable database growth rate as a function of deletion detection probability $\alpha$ and showed that a nonzero $\alpha$ can significantly improve the achievable rate. Then assuming no initial deletion location information ($\alpha=0$), we have proposed an algorithm for detecting deletion locations when a batch of $B$ correctly-matched seed rows are given. We have found that in order for this algorithm to guarantee a non-zero detection probability, we need $B=O(\log n)=O(\log\log m)$. Our ongoing work considers matching at the database level rather than matching each row separately, potentially leading to higher achievable rates.

\bibliography{references}

\begin{thebibliography}{10}
\providecommand{\url}[1]{#1}
\csname url@samestyle\endcsname
\providecommand{\newblock}{\relax}
\providecommand{\bibinfo}[2]{#2}
\providecommand{\BIBentrySTDinterwordspacing}{\spaceskip=0pt\relax}
\providecommand{\BIBentryALTinterwordstretchfactor}{4}
\providecommand{\BIBentryALTinterwordspacing}{\spaceskip=\fontdimen2\font plus
\BIBentryALTinterwordstretchfactor\fontdimen3\font minus
  \fontdimen4\font\relax}
\providecommand{\BIBforeignlanguage}[2]{{%
\expandafter\ifx\csname l@#1\endcsname\relax
\typeout{** WARNING: IEEEtran.bst: No hyphenation pattern has been}%
\typeout{** loaded for the language `#1'. Using the pattern for}%
\typeout{** the default language instead.}%
\else
\language=\csname l@#1\endcsname
\fi
#2}}
\providecommand{\BIBdecl}{\relax}
\BIBdecl

\bibitem{naini2015you}
F.~M. {Naini}, J.~{Unnikrishnan}, P.~{Thiran}, and M.~{Vetterli}, ``Where you
  are is who you are: User identification by matching statistics,'' \emph{IEEE
  Trans. Inf. Forensics Security}, vol.~11, no.~2, pp. 358--372, 2016.

\bibitem{datta2012provable}
A.~Datta, D.~Sharma, and A.~Sinha, ``Provable de-anonymization of large
  datasets with sparse dimensions,'' in \emph{International Conference on
  Principles of Security and Trust}.\hskip 1em plus 0.5em minus 0.4em\relax
  Springer, 2012, pp. 229--248.

\bibitem{narayanan2008robust}
A.~{Narayanan} and V.~{Shmatikov}, ``Robust de-anonymization of large sparse
  datasets,'' in \emph{2008 IEEE Symposium on Security and Privacy}, 2008, pp.
  111--125.

\bibitem{sweeney1997weaving}
L.~Sweeney, ``Weaving technology and policy together to maintain
  confidentiality,'' \emph{The Journal of Law, Medicine \& Ethics}, vol.~25,
  no. 2-3, pp. 98--110, 1997.

\bibitem{takbiri2018matching}
N.~Takbiri, A.~Houmansadrand, D.~Goeckel, and H.~Pishro-Nik, ``Matching
  anonymized and obfuscated time series to users’ profiles,'' \emph{IEEE
  Trans. Inf. Theory}, vol.~65, no.~2, pp. 724--741, 2018.

\bibitem{shirani8849392}
F.~{Shirani}, S.~{Garg}, and E.~{Erkip}, ``A concentration of measure approach
  to database de-anonymization,'' in \emph{2019 IEEE International Symposium on
  Information Theory (ISIT)}, 2019, pp. 2748--2752.

\bibitem{cullina}
D.~{Cullina}, P.~{Mittal}, and N.~{Kiyavash}, ``Fundamental limits of database
  alignment,'' in \emph{2018 IEEE International Symposium on Information Theory
  (ISIT)}, 2018, pp. 651--655.

\bibitem{cheraghchi2020overview}
M.~Cheraghchi and J.~Ribeiro, ``An overview of capacity results for
  synchronization channels,'' \emph{IEEE Trans. Inf. Theory}, 2020.

\bibitem{diggavi1603788}
S.~{Diggavi} and M.~{Grossglauser}, ``On information transmission over a finite
  buffer channel,'' \emph{IEEE Trans. Inf. Theory}, vol.~52, no.~3, pp.
  1226--1237, 2006.

\bibitem{chvatal1975longest}
V.~Chvatal and D.~Sankoff, ``Longest common subsequences of two random
  sequences,'' \emph{Journal of Applied Probability}, pp. 306--315, 1975.

\bibitem{cover2006elements}
T.~M. Cover, \emph{Elements of Information Theory}.\hskip 1em plus 0.5em minus
  0.4em\relax John Wiley \& Sons, 2006.

\bibitem{shirani2017seeded}
F.~Shirani, S.~Garg, and E.~Erkip, ``Seeded graph matching: Efficient
  algorithms and theoretical guarantees,'' in \emph{2017 51st Asilomar
  Conference on Signals, Systems, and Computers}, 2017, pp. 253--257.

\bibitem{fishkind2019seeded}
D.~Fishkind, S.~Adali, H.~Patsolic, L.~Meng, D.~Singh, V.~Lyzinski, and
  C.~Priebe, ``Seeded graph matching,'' \emph{Pattern Recognition}, vol.~87,
  pp. 203--215, 2019.

\bibitem{mitzenmacher2009survey}
M.~Mitzenmacher, ``A survey of results for deletion channels and related
  synchronization channels,'' \emph{Probability Surveys}, vol.~6, pp. 1--33,
  2009.

\end{thebibliography}
\bibliographystyle{IEEEtran}
\end{document}